\documentclass[11pt,letterpaper]{article}
\usepackage{dblfloatfix}
\usepackage{bbm}
\usepackage{blindtext}
\usepackage{algpseudocode}
\usepackage{amsmath}
\usepackage{amssymb}
\usepackage{amsthm}
\usepackage{color}
\usepackage{cite}
\usepackage{graphicx}
\graphicspath{{./}{Figs/}}
\usepackage{subcaption}
\usepackage{mathrsfs}
\usepackage{verbatim}
\usepackage{indentfirst}
\usepackage[titletoc,toc,title]{appendix}
\usepackage{appendix}
\usepackage{url}
\usepackage{epsfig,endnotes,amsthm,changepage}
\usepackage{authblk}

\usepackage{epstopdf,lipsum,etoolbox}

\usepackage{graphicx}
\usepackage{amsfonts}
\usepackage{setspace}
\usepackage{cite}
\usepackage{array}
\usepackage{mdwmath}
\usepackage{mdwtab}
\usepackage{mdwtab, stmaryrd}
\usepackage{times}
\usepackage{epsfig}
\usepackage{latexsym}
\usepackage{epstopdf}
\usepackage{verbatim}
\usepackage{units}
\usepackage{amsthm}
\usepackage{mdwlist}
\usepackage{dblfloatfix}
\usepackage{blindtext}
\usepackage{array}
\newcolumntype{P}[1]{>{\centering\arraybackslash}p{#1}}
\usepackage{tabu}

\AppendGraphicsExtensions{.pdf}
\usepackage{epstopdf,lipsum,etoolbox}

\newtheorem{definition}{Definition}
\newtheorem{theorem}{Theorem}

\newtheorem*{proposition*}{Proposition}

\newtheorem{remark}{Remark}

\allowdisplaybreaks

\usepackage{dsfont}

\begin{document}

\title{Wireless Federated Learning with Local Differential Privacy}

 \author{{Mohamed~Seif  \quad Ravi~Tandon \quad Ming~Li}\\
 Department of Electrical and Computer Engineering\\
 University of Arizona\\
 Email: \{\textit{mseif, tandonr, lim}\}@email.arizona.edu
 }

\maketitle

\begin{abstract}
In this paper, we study the problem of federated learning (FL) over a wireless channel, modeled by a Gaussian multiple access channel (MAC), subject to local differential privacy (LDP) constraints. We show that the superposition nature of the wireless channel provides a dual benefit of bandwidth efficient gradient aggregation, in conjunction with strong LDP guarantees for the users. We propose a private wireless gradient aggregation scheme, which shows that when aggregating gradients from $K$ users, the privacy leakage per user scales as $\mathcal{O}\big(\frac{1}{\sqrt{K}} \big)$ compared to orthogonal transmission in which the privacy leakage scales as a constant. We also present analysis for the convergence rate of the proposed private FL aggregation algorithm and study the tradeoffs between wireless resources, convergence, and privacy.
\end{abstract}




\vspace{-4pt}
\section{Introduction}
\label{sec:introduction}
\footnote{This work was supported by US NSF through grants CAREER 1651492, CNS 1715947, and by the Keysight Early Career Professor Award.}

 Federated learning (FL) \cite{mcmahan2017communication} is a framework that enables multiple users to jointly train a learning model. In prototypical FL, a central server interacts with multiple users to train a ML model in an iterative manner as follows: users compute gradients for the ML model on their local data sets, and gradients are subsequently exchanged for model updates. There are several motivating factors behind the surging popularity of FL: a) centralized approaches can be inefficient in terms of storage/computation, and FL provides natural parallelization for training, and can leverage increasing computational power of devices and b) local data at each user is never shared, but only gradient computations from each user are collected. Despite the fact that in F-ML, local data is never shared by a user, even exchanging gradients in a raw form can leak information, as shown in recent works \cite{shokri2017membership, hayes2019logan, melis2019exploiting}.
 
 
 Motivated by these factors, there has been a recent surge in designing F-ML algorithms with rigorous privacy guarantees. Differential privacy (DP) \cite{Dwork20061} has been adopted a \textit{de facto} standard notion for private data analysis and aggregation. Within the context of FL, the notion of local differential privacy (LDP) is more suitable in which a  user can locally perturb and disclose the data to an \textit{untrusted} data curator/aggregator \cite{joseph2018local}. LDP has been already adopted and used in current applications, including Google's RAPPOR \cite{fanti2016building} for website browsing history aggregation, and by Microsoft for privately collecting telemetry data \cite{ding2017collecting}.  In the literature, there has been several research efforts to design FL algorithms satisfying LDP \cite{triastcyn2019federated, geyer2017differentially, bagdasaryan2019differential, wu2019distributed, choudhury2019differential, wei2019performance, agarwal2018cpsgd}. 
 While LDP provides stronger privacy guarantees (compared to a centralized solution), this comes at the cost of lower utility. In particular, to achieve the same level of privacy attained by a centralized solution, significant higher amount of noise/perturbation
  is needed \cite{cormode2018privacy, wang2018empirical, bassily2019linear, bassily2015local, bassily2017practical}.

Another parallel recent trend is to study the feasibility of FL over wireless channels. As the prototypical computation for FL training involves gradient aggregation from multiple users, the superposition property of the wireless channel can naturally support this operation much more efficiently. This has led to several recent works \cite{amiri2019machine, zhu2018low, zeng2019energy, yang2018federated, wang2019adaptive, amiri2019federated, sery2019analog, sery2019sequential, abad2019hierarchical, khan2019federated, amiri2019over} under the umbrella of FL at the wireless edge, where distributed users interact with a parameter server (PS) over a shared wireless medium for training ML models. Several methodologies have been proposed to study wireless FL, which can be broadly categorized into either digital or analog aggregation schemes. In digital schemes, quantized gradients from each user are individually transmitted to the PS using orthogonal transmission. For analog schemes, on the other hand, the gradient computations are rescaled and transmitted directly over the air by all users simultaneously. The superposition nature of the wireless medium makes analog schemes more bandwidth efficient compared to digital ones.   

In this paper, we focus on the following question: \textit{Can the superposition property of wireless also be beneficial for privacy? If yes, how can we optimally utilize the wireless resources, and what are the tradeoffs between convergence of F-ML training, wireless resources and privacy?}

\noindent \textbf{Main Contributions}:  In this paper, we consider the problem of FL training over a flat-fading Gaussian multiple access channel (MAC), subject to LDP constraints. We propose and study analog aggregation schemes, in which each user transmits a linear combination of a) local gradients and b) artificial Gaussian noise, subject to power constraints. The local gradients are processed as a function of the channel gains to \textit{align} the resulting gradients at the PS, whereas the artificial noise parameters are selected to satisfy the privacy constraints.  We show that  the privacy level per user scales as $\mathcal{O}\big(\frac{1}{\sqrt{K}} \big)$ compared to orthogonal transmission in which the privacy leakage scales as a constant.
We also provide the privacy-convergence trade-offs for smooth and convex loss functions through convergence analysis of the distributed gradient descent algorithm. We show that the training error decreases as the number of users increases and converges to the centralized algorithm where all points are available at the PS. To the best of our knowledge, this is the first result on wireless FL with LDP constraints. 

\section{System Model \& Problem Statement}\label{system_model_section}

  \begin{figure}[t]
        \centering
        \includegraphics[width=0.55\linewidth]{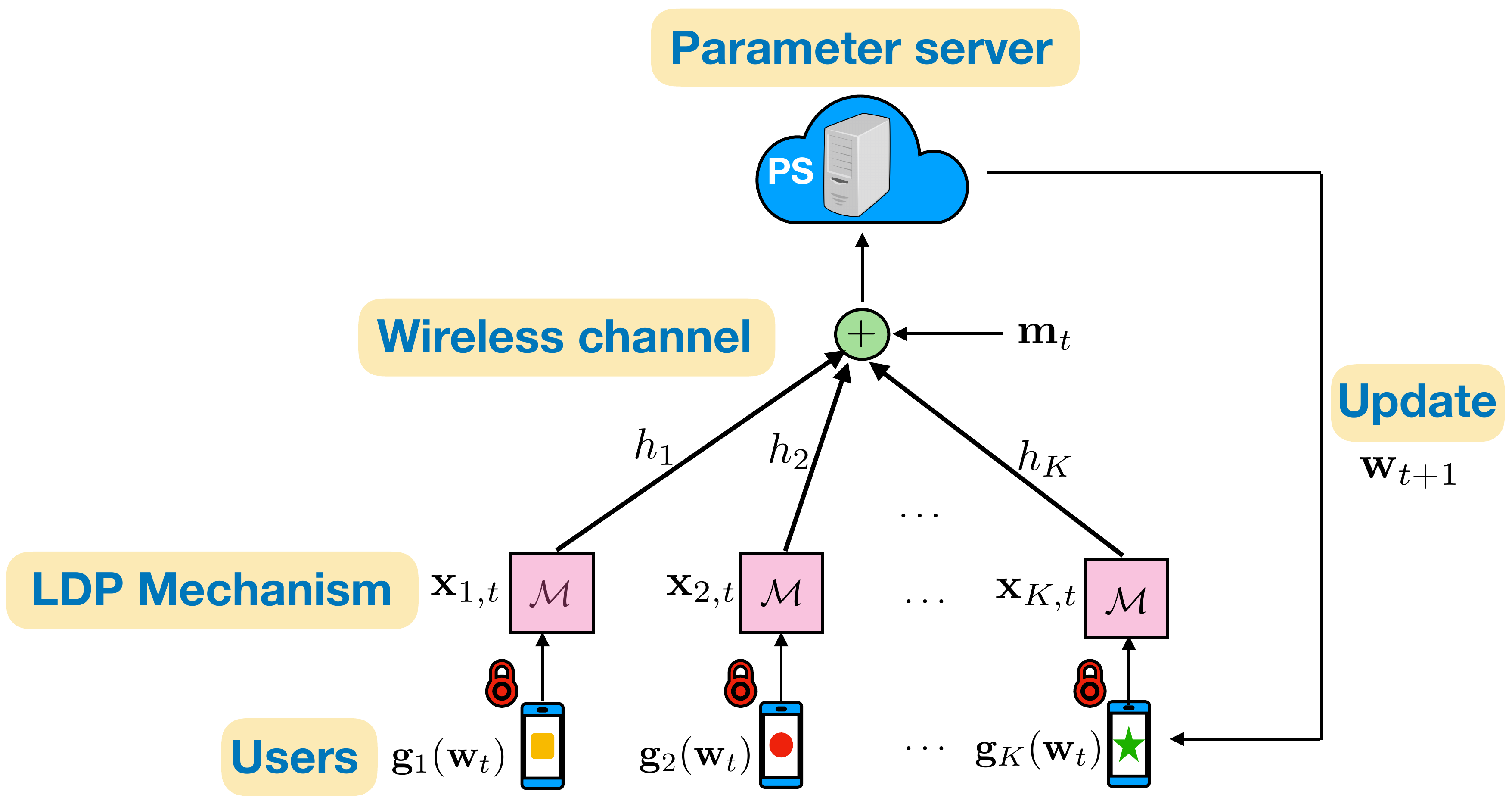}
        \caption{Illustration of the private wireless FL framework: Users collaborate with the PS to jointly train a machine learning model over a Gaussian MAC. The interaction between the users and the PS must satisfy local differential privacy (LDP) constraints for each user.  }
        \label{fig:my_label}
        \vspace{-15pt}
    \end{figure}
    
\noindent \textit{Wireless Channel Model}: We consider a single-antenna wireless FL system with $K$ users and a central PS  as shown in Fig.  \ref{fig:my_label}. The input-output relationship at time $i$ is
\begin{align}\label{eq:systemmodel}
      y(i) =  \sum_{k=1}^{K} h_{k} x_{k}(i)  + {m}(i),
\end{align}
where $x_{k}(i)$ is the  signal transmitted by  user $k$ at time $i$, and $y(i)$ is the received signal at the PS. Here, $h_{k} = |h_{k}| e^{j \phi_{k}}$ is the complex valued channel coefficient between the $k$-th user and the PS, and and $m(i)$ is the independent additive zero-mean unit-variance (AWGN) Gaussian noise.
The channel coefficients are assume to be time invariant, and each user can transmit subject to maximum power constraint of $P_k$. Each user is assumed to know its local channel gains, whereas we assume that the PS has global channel state information. 


\noindent \textit{Federated Learning Problem}:
Each user $k$ has a private local dataset $\mathcal{D}_{k}$ of size $|\mathcal{D}_{k}| $ data points, denoted as $\mathcal{D}_{k} = \{(\mathbf{u}_{i}^{(k)}, v_{i}^{(k)})\}_{i=1}^{{|\mathcal{D}_{k}|}}$, where $\mathbf{u}_{i}^{(k)}$ is the $i$-th data point and $v_{i}^{(k)}$ is the corresponding label at user $k$. Users communicate with the PS through the Gaussian MAC described above in order to  train a model by minimizing the loss function $F(\mathbf{w})$, i.e., 
\begin{align}
    \mathbf{w}^{*} = \text{arg} \min_{\mathbf{w}} F(\mathbf{w}) \triangleq \frac{1}{|\mathcal{D}_{\textsf{total}}|} \sum_{k=1}^{K} \sum_{i = 1}^{|\mathcal{D}_{k}|  } f_{k}( (\mathbf{u}_{i}^{(k)}, v_{i}^{(k)}); \mathbf{w}), \nonumber 
\end{align}
where $\mathbf{w} \in \mathds{R}^{d}$ is the parameter vector to be optimized, $f_{k}(\cdot)$ is the loss function for user $k$, and $\mathcal{D}_{\textsf{total}} = \cup_{k=1}^{K}  \mathcal{D}_{k}$ denotes the entire dataset used for  training. The minimization of $F(\mathbf{w})$ is carried out iteratively through a distributed gradient descent (GD) algorithm. More specifically, in the $t$-th training iteration, the PS broadcasts the global parameter vector $\mathbf{w}_{t}$ from the last iteration to all users. Each user $k$  computes his local gradient over the local $|\mathcal{D}_{k}|$ data points, i.e., $\mathbf{g}_{k}(\mathbf{w}_{t}) = \frac{1}{|\mathcal{D}_{k}|} \sum_{i= 1}^{|\mathcal{D}_{k}|}  \nabla f_{k}((\mathbf{u}_{i}^{(k)}, v_{i}^{(k)}); \mathbf{w})$ and sends back the computed gradient to the PS. For the scope of this paper, we assume that $|\mathcal{D}_{k}| = |\mathcal{D}|$, therefore $ |\mathcal{D}_{\textsf{total}}| = K |\mathcal{D}|$. The global parameter $\mathbf{w}_{t}$ is updated according to
\begin{align}
    \mathbf{w}_{t+1} = \mathbf{w}_{t} - \eta_{t} \frac{1}{K} \sum_{k=1}^{K} \mathbf{g}_{k}(\mathbf{w}_{t}),
\end{align}
where $\eta_{t}$ is the learning rate of the distributed GD algorithm at iteration $t$. The iteration process continues until convergence.

In addition, the gradient descent (GD) algorithm for wireless FL should also satisfy local differential privacy (LDP) constraints for each user, as defined next. 
\begin{definition} ($(\epsilon, \delta)$-LDP \cite{dwork2014algorithmic}) A randomized mechanism $\mathcal{M}: \mathcal{X} \rightarrow \mathds{R}^{d}$ is $(\epsilon, \delta)$-LDP if for any pair $x, x' \in \mathcal{X}$ and any measurable subset $\mathcal{O} \subseteq \text{Range}(\mathcal{M})$, we have 
\begin{align}
    \operatorname{Pr}(\mathcal{M}(x) \in \mathcal{O}) \leq e^{\epsilon}  \operatorname{Pr}(\mathcal{M}(x') \in \mathcal{O}) + \delta.
\end{align}
The case of $\delta = 0$ is called pure $\epsilon$-LDP.
\end{definition}
\noindent \textbf{Problem Statement.} The main goal of this paper is to explore the benefits of wireless gradient aggregation for privacy in FL. In addition, we investigate tradeoffs between the convergence rate of GD, wireless channel conditions and resources (such as power, SNR), subject to the privacy budgets of the users.

\section{Main Results \& Discussions}
In this Section, we present a general gradient aggregation scheme for wireless FL, where each user transmits a linear combination of its local gradients and artificial noise. We then specialize this scheme in which the part of transmission containing gradients are designed in a manner so that this component is aligned at the PS. We 
analyze this scheme and obtain the privacy leakage under LDP for each user, as a function of the wireless channel conditions, and the transmission parameters. Finally, we present the convergence rate of the private FL algorithm, and maximize the convergence rate by optimizing the local perturbations of each user for privacy.
\subsection{FL Transmission Scheme over Gaussian MAC}
\vspace{-5pt}
The overall FL scheme consists of $T$ training iterations, where each iteration comprises of $d$ uses of the wireless channel described in \eqref{eq:systemmodel}. At each iteration $t$, each user $k$ transmits the computed gradient vector $\mathbf{g}_{k}(\mathbf{w}_{t}) \in \mathds{R}^{d}$ together with additive Gaussian noise for privacy. In particular, the transmitted signal of user $k$ at iteration $t$ is given as:
\begin{align}
    \mathbf{x}_{k, t} & =  e^{-j \phi_{k}} \left[\underbrace{\frac{\sqrt{\alpha_{k} P_{k}}}{L} {\mathbf{g}_{k}(\mathbf{w}_{t}) }}_{\textsf{local gradient estimate}} + \underbrace{\sqrt{\beta_{k} P_{k}} \mathbf{n}_{k,t}}_{\textsf{local perturbation}} \right] \label{eq:inputsignal}
\end{align}
Here, each user $k$ performs local phase correction (i.e., input is multiplied by $e^{-j\phi_k}$) so that the received channel coefficient is non-negative, i.e., $|h_{k}|$. We assume that the gradient vectors have a bounded norm, i.e., $\|\mathbf{g}_{k}(\mathbf{w}_{t})\|_{2} \leq L, \forall k$, and normalize the gradient vector by $L$.  Here, $\alpha_k \in [0,1]$ denotes the fraction of power dedicated to the gradient vector $\mathbf{g}_{k}(\mathbf{w}_{t})$, whereas $\beta_k \in [0,1-\alpha_k]$ is the fraction of power dedicated to artificial Gaussian noise $\mathbf{n}_{k, t}$, whose elements are i.i.d., and drawn from $\mathcal{N}(0,1)$. These parameters satisfy $\alpha_k + \beta_k \leq 1$ so that the maximum power constraint of $P_k$ is satisfied. From \eqref{eq:systemmodel} and \eqref{eq:inputsignal}, the received signal at the PS can be written as:
\begin{align}
    \mathbf{y}_{t} & = \sum_{k=1}^{K} |h_{k}| \left[ \frac{\sqrt{\alpha_{k} P_{k}}}{L} {\mathbf{g}_{k}(\mathbf{w}_{t}) } + \sqrt{\beta_{k} P_{k}} \mathbf{n}_{k,t}  \right] + \mathbf{m}_{t} \nonumber \\
    & =  \underbrace{\sum_{k=1}^{K} |h_{k}|  \frac{\sqrt{\alpha_{k} P_{k}}}{L} {\mathbf{g}_{k}(\mathbf{w}_{t}) }}_{\textsf{aggregated gradient at PS}} +  \underbrace{\sum_{k=1}^{K} |h_{k}|  \sqrt{\beta_{k} P_{k}} \mathbf{n}_{k,t} + \mathbf{m}_{t}}_{\textsf{aggregated noise at PS}},\label{eq:output}
    \end{align}
 where  $\mathbf{m}_{t} \in \mathds{R}^{d}$ is the independent Gaussian noise, whose elements are i.i.d. drawn from $\mathcal{N}(0, \sigma_{m}^{2})$.
 In order to carry out the summation of the local gradients over-the-air, and receive an unbiased estimate of the true aggregated gradient, all users pick the coefficients $\alpha_k$s in order to align their transmitted local gradient estimates. Specifically, user $k$ picks $\alpha_k$ so that 
\begin{align}
       \frac{|h_{k}| \sqrt{\alpha_{k} P_{k}}}{L}  = c, \forall k,\label{eq:alpha11}
\end{align}
where $c$ is a constant. From \eqref{eq:alpha11}, we obtain $\alpha_{k} = \frac{c^{2} L^{2}}{|h_{k}|^{2} P_{k}}$, and using the fact that $\alpha_k \leq 1$, for all $k$, we can upper bound the constant $c$ as follows:
    $c \leq \frac{\sqrt{\min_{j} |h_{j}|^{2} P_{j}}}{L}$.
To maximize the signal power of the aligned gradient, we choose $c$ to match this upper bound, i.e., 
\begin{align}
c = \frac{\sqrt{\min_{j} |h_{j}|^{2} P_{j}}}{L}.\label{eq:cvalue}
\end{align}
Plugging this back in \eqref{eq:alpha11}, we obtain the choice of $\alpha_{k}$ as 
\begin{align}
    \alpha_{k} = \frac{ \min_{j} |h_{j}|^{2} P_{j} }{ |h_{k}|^{2} P_{k} }.
\end{align}
The above choice shows that alignment of gradients is effectively limited by the user with the worst effective SNR, i.e., $\min_{j} |h_{j}|^{2} P_{j}$. For the alignment scheme described above, the received signal by the PS in iteration $t$ in \eqref{eq:output} simplifies to:
\begin{align}
   \mathbf{y}_{t} = c \sum_{k=1}^{K} {\mathbf{g}_{k}(\mathbf{w}_{t})} +  \sum_{k=1}^{K} |h_{k}|  \sqrt{\beta_{k} P_{k}} \mathbf{n}_{k,t} + \mathbf{m}_{t}. \label{eq:ytbeforepreprocessing}
\end{align}
The PS  subsequently performs post-processing on $\mathbf{y}_{t}$ as follows:
\begin{align}
    &\hat{\mathbf{g}}_{t}  = \frac{1}{K c} \times \mathbf{y}_{t} \nonumber \\ 
    &=  \underbrace{\frac{1}{K} \sum_{k=1}^{K} {\mathbf{g}_{k}(\mathbf{w}_{t})}}_{\nabla F(\mathbf{w}_{t})} + \underbrace{ \frac{1}{Kc} \times  \left[ \sum_{k=1}^{K} |h_{k}|  \sqrt{\beta_{k} P_{k}} \mathbf{n}_{k,t} +  \mathbf{m}_{t} \right]}_{\mathbf{z}_{t}}, \label{eq:postpreprocessing}
\end{align}
where  $\mathbf{z}_{t} \sim \mathcal{N}(0,  \sigma_{z}^{2} \mathbf{I}_{d})$ is the effective noise at the PS, and $\sigma_{z}^{2} = \frac{1}{K^{2} c^{2}} \left[ \sum_{k=1}^{K} |h_{k}|^{2} \beta_{k} P_{k} + \sigma_{m}^{2}\right]$. Thus, we can write $\hat{\mathbf{g}}_{t} = \nabla F(\mathbf{w}_{t}) + \mathbf{z}_{t}$. As $\mathbf{z}_{t}$ is zero mean, $\hat{\mathbf{g}}_{t}$ is an unbiased estimate of $\nabla F(\mathbf{w}_{t})$, with variance of $\hat{\mathbf{g}}_{t}$ being equal to $\sigma_{z}^{2}$.

\subsection{Local Differential Privacy Analysis}
We next analyze the privacy level achieved by the transmission scheme for each user, as per the definition of LDP. Recall, that the local perturbation noise is drawn from Gaussian distribution. This well-known technique is known as Gaussian mechanism and can provide rigorous privacy guarantees based on LDP, as defined next.

\begin{definition} (Gaussian Mechanism - Appendix A of \cite{dwork2014algorithmic}) Suppose a user wants to release a function $f(X)$ of an input $X$ subject to $(\epsilon, \delta)$-LDP. The Gaussian release mechanism is defined as:
\begin{align}
M(X) \triangleq f(X) + \mathcal{N}(0, \sigma^{2} \mathbf{I}).
\end{align}
If the sensitivity of the function is bounded by $\Delta_f$, i.e., $\| f(x) - f(x')\|_{2}\leq \Delta_f$, $\forall x, x'$, then for any $\delta \in (0,1]$, Gaussian mechanism satisfies $(\epsilon, \delta)$-LDP, where 
\begin{align}
    \epsilon = \frac{\Delta_{f}}{\sigma} \sqrt{2 \log \frac{1.25}{\delta}}. \label{gaussian_noise_add}
\end{align}
\end{definition} 
In the next Theorem, we make use of the above result, and present the per-user privacy achieved by the proposed wireless FL scheme as a function of the noise power allocation parameters $\{\beta_k\}_{k=1}^{K}$, transmit powers $\{P_k\}_{k=1}^{K}$, and the channel coefficients $\{h_k\}_{k=1}^{K}$. 

\begin{theorem} 
For each user $k$, the proposed transmission scheme achieves $(\epsilon_{k}, \delta)$-LDP  per iteration, where
\begin{align}
    \epsilon_{k} & =  \frac{2  \sqrt{\min_{j} |h_{j}|^{2} P_{j}}}{\sqrt{\sum_{k=1}^{K} |h_{k}|^{2} \beta_{k} P_{k} + \sigma_{m}^{2}}}  \sqrt{2 \log \frac{1.25}{\delta} }.
\end{align}\label{theorem1:privacy}
\end{theorem}
\begin{proof}
The final received signal at the PS from \eqref{eq:ytbeforepreprocessing} can be expressed as:
$    \mathbf{y}_{t} = c \sum_{k=1}^{K} {\mathbf{g}_{k}(\mathbf{w}_{t})} + K c \mathbf{z}_t$.
We first observe that the variance of the effective Gaussian noise, i.e., variance of $K c \mathbf{z}_t$ is 
$\sigma^{2} =  \sum_{k=1}^{K} |h_{k}|^{2} \beta_{k} P_{k} $  
$ + \sigma_{m}^{2}$. In order to invoke the result of the Gaussian mechanism, we next obtain a bound on the sensitivity for user $k$. To bound the local sensitivity of $c \sum_{k=1}^{K} {\mathbf{g}_{k}(\mathbf{w}_{t})}$, consider any two different local datasets $\mathcal{D}_{k}$ and $\mathcal{D}'_{k}$ at user $k$, while fixing the datasets (and thus the gradients) of the remaining $(K-1)$ users. The local sensitivity of user $k$ can then be bounded as
\begin{align}
    \Delta_{k} 
    & = \max_{\mathcal{D}_{k}, \mathcal{D}'_{k}}
||  \mathbf{y}_t - \mathbf{y}^{'}_t ||_2 = \max_{\mathcal{D}_{k}, \mathcal{D}'_{k}}
||  c (\mathbf{g}_{k}(\mathbf{w}_{t}) - \mathbf{g}'_{k}(\mathbf{w}_{t})) ||_2 \nonumber\\
& \leq  c\max_{\mathcal{D}_{k}, \mathcal{D}'_{k}}
||  \mathbf{g}_{k}(\mathbf{w}_{t})||_2 + ||\mathbf{g}'_{k}(\mathbf{w}_{t})||_2 \overset{(a)} \leq 2cL\nonumber\\ & \overset{(b)} = 2  \sqrt{\min_{j} |h_{j}|^{2} P_{j}}, \label{eq:sensitivity}
\end{align}
where in step (a), we used the fact that $\| {\mathbf{g}_{k}(\mathbf{w}_{t})} \|_{2}   \leq L, \forall k$, and (b) follows from \eqref{eq:cvalue}. Hence, using the sensitivity bound in \eqref{eq:sensitivity} together with the variance $\sigma^{2} =  {\sum_{k=1}^{K} |h_{k}|^{2} \beta_{k} P_{k} + \sigma_{m}^{2}}$ in \eqref{gaussian_noise_add}, we arrive 
at the proof of Theorem \ref{theorem1:privacy}. 





\end{proof}


\begin{remark} From Theorem \ref{theorem1:privacy}, we can observe the privacy benefits of wireless gradient aggregation. 
We can further upper bound the achievable $\epsilon_k$ in Theorem \ref{theorem1:privacy} as follows:
\begin{align}
        \epsilon_{k} 
        & =  \frac{2  \sqrt{\min_{j} |h_{j}|^{2} P_{j}}}{\sqrt{\sum_{k=1}^{K} |h_{k}|^{2} \beta_{k} P_{k} + \sigma_{m}^{2}}}  \sqrt{2 \log \frac{1.25}{\delta} } \nonumber \\
        & \leq \frac{2  \sqrt{\min_{j} |h_{j}|^{2} P_{j}}}{\sqrt{\sum_{k=1}^{K} |h_{k}|^{2} \beta_{k} P_{k} }}  \sqrt{2 \log \frac{1.25}{\delta} } \nonumber\\
        &\leq  \frac{1}{\sqrt{K}} \times  \frac{2  \sqrt{\min_{j} |h_{j}|^{2} P_{j}}}{\sqrt{ \min_{k}|h_{k}|^{2} \beta_{k} P_{k} }}  \sqrt{2 \log \frac{1.25}{\delta} }, \nonumber
\end{align}
which shows that asymptotically, the per-user privacy level behaves like $\mathcal{O}(1/\sqrt{K})$. In contrast, privacy achieved by
orthogonal transmission can be shown to be:
\begin{align}
    \epsilon^{\text{Orthogonal}}_{k}  =   \frac{2|h_k| \sqrt{\alpha_{k}P_k}}{\sqrt{|h_k|^{2}\beta_{k}P_k + \sigma_{m}^{2}}}  \sqrt{2 \log \frac{1.25}{\delta} },  
\end{align}
which scales as a constant, and does not decay with $K$. 
\end{remark}


\begin{remark} While Theorem \ref{theorem1:privacy} shows the per-iteration leakage, we can use advanced composition results for LDP using the Gaussian mechanism to obtain the total privacy leakage when the wireless FL algorithm is used for $T$ iterations. Using existing results in \cite{dwork2010boosting}, it can be readily shown that the total leakage over $T$ iterations (per-user) of the proposed scheme is $(\epsilon_{k}^{(T)}, T \delta + \delta^{'})$-LDP for $\delta^{'} \in (0, 1]$ where,  
\begin{align}
    \epsilon_k^{(T)}= \sqrt{2 T \log(1/\delta^{'})} \epsilon_{k} + T \epsilon_{k} (e^{\epsilon_{k}} -1).
\end{align}
We illustrate the total per-user privacy leakage as a function of $K$, the number of users in Fig. \ref{fig:privacyscaling} for various values of $T$. As is clearly evident, the leakage provided by wireless FL goes asymptotically to $0$ as $K\rightarrow \infty$. 
\end{remark}

\begin{figure}[t]
    \centering
    \includegraphics[width=0.55\linewidth]{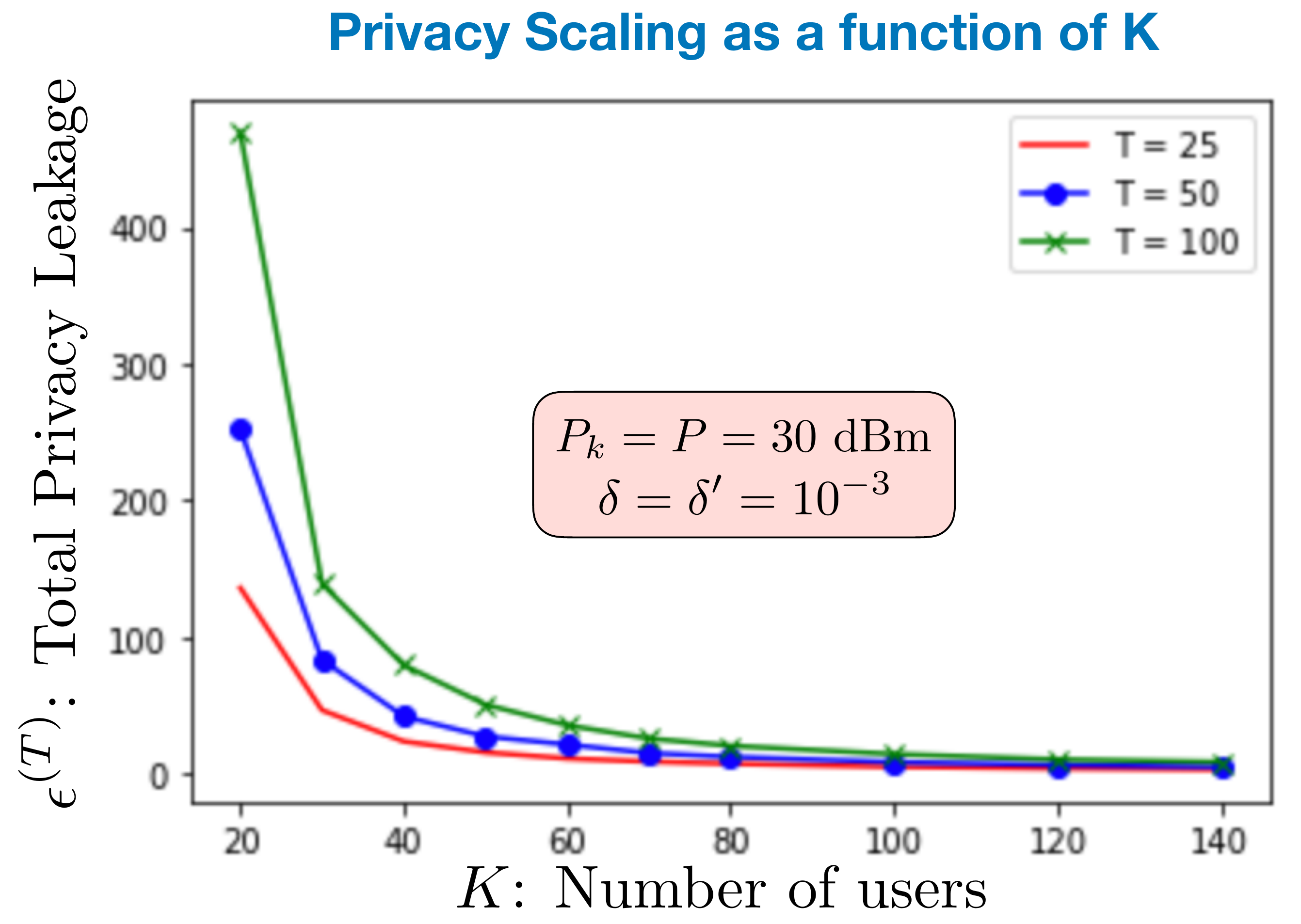}
    \vspace{-5pt}
    \caption{Total per-user privacy leakage as a function of $K$, number of users for different values of $T$, the number of training iterations.}
    \label{fig:privacyscaling}
    \vspace{-10pt}
\end{figure}

\subsection{Convergence rate of private FL}
We next analyze the performance of  private wireless FL under the assumption that the global loss function $F(\mathbf{w})$ is smooth and strongly convex. Due to privacy requirements and noisy nature of wireless channel, the convergence rate is penalized  as shown in the following Theorem. 
\begin{theorem}\label{theorem_1_convergence}
Suppose the loss function $F$ is $\lambda$-strongly convex and $\mu$-smooth with respect to $\mathbf{w}^{*}$. Then, for a learning rate $\eta_{t} = 1 / \lambda t$ and a number of iterations $T$, the convergence rate of the private wireless FL algorithm is 
\begin{align}
    & \mathds{E} \left[ F(\mathbf{w}_{T}) \right] - F(\mathbf{w}^{*}) 
     \leq  \frac{2 \mu}{ \lambda^{2} T} \times  \left[ L^{2} +
 \frac{d}{K^{2} c^{2}} \left[ \sum_{k=1}^{K} |h_{k}|^{2} \beta_{k} P_{k} + \sigma_{m}^{2} \right] \right]. \label{convergence_equation}
\end{align}
\end{theorem}

Theorem \ref{theorem_1_convergence} is proved in Appendix I.   
We next show that artificial noise parameters $\{\beta_k\}_{k=1}^{K}$ can be optimized to maximize the convergence rate in (\ref{convergence_equation})  while satisfying a desired privacy level $(\epsilon_{k}, \delta)$-LDP at each user. 

\begin{theorem}\label{optimized_convergence}
The optimized convergence rate of the private wireless FL algorithm is given as follows:
\begin{align}
    & \mathds{E} \left[ F(\mathbf{w}_{T}) \right] - F(\mathbf{w}^{*})  
     \leq   \frac{2 \mu }{ \lambda^{2} T} \times  \left[ L^{2} +
 \frac{d}{K^{2} c^{2}} \left[ \sum_{k=1}^{K} Z_{k} + \sigma_{m}^{2} \right] \right], \label{convergence_equation_2}
\end{align}
where $  Z_{k}  = \min \left[ \lambda_{k}, (\Psi - \sum_{i = 1}^{k-1} U_{i})^{+} \right]$ where $(a)^{+} \triangleq \max (0, a) $, $\lambda_{k} = |h_{k}|^{2} P_{k} (1-\alpha_{k})$, \\
$ \Psi = \max_{i} \frac{8  \min_{j} |h_{j}|^{2} P_{j} }{\epsilon_{i}^{2}} \log \frac{1.25}{\delta} - \sigma_{m}^{2}$, and $U_{i} = |h_{i}|^{2} P_{i} \beta_{i} $. 
\end{theorem}

\begin{proof}
Maximizing the convergence rate in (\ref{convergence_equation}) is equivalent to minimizing the term that depends on $\{\beta_{k}\}_{k=1}^{K}$. Therefore, we solve the following optimization problem:
\begin{align}
      &\min_{\{\beta_{i}\}_{k=1}^{K} }  \sum_{k=1}^{K} |h_{k}|^{2} \beta_{k} P_{k} ~~ \text{such that~~} 0\leq \beta_{k} \leq 1-\alpha_{k}, \forall k, \nonumber \\ 
    & ~~\&~~  \sum_{k=1}^{K}  |h_{k}|^{2} \beta_{k} P_{k} \geq \frac{8  \min_{j} |h_{j}|^{2}  P_{j}}{\epsilon_{k}^{2}} \log \frac{1.25}{\delta } - \sigma_{m}^{2}.  \nonumber  \label{lp_problem}
\end{align}
For given target privacy levels $\{\epsilon_{k}\}_{k=1}^{K}$, this is feasible when
\begin{align}
    \sum_{k = 1}^{K} \underbrace{|h_{k}|^{2} P_{k} (1 - \alpha_{k})}_{\lambda_{k}} \geq {\max_{i} \frac{8  \min_{j} |h_{j}|^{2} P_{j} }{\epsilon_{i}^{2}} \log \frac{1.25}{\delta} - \sigma_{m}^{2}}. \nonumber
\end{align}
 We design $\beta_{k}, \forall k$ as follows:
 \begin{align}
    \beta_{k} = \frac{Z_{k}}{|h_{k}|^{2} P_{k}}, k = 1, \cdots, K.
\end{align}
where 
$Z_{k} = \min \left[ \lambda_{k}, (\Psi - \sum_{i = 1}^{k-1} U_{i})^{+} \right], k = 1, \cdots, K$, $\Psi = \max_{i} \frac{8  \min_{j} |h_{j}|^{2} P_{j} }{\epsilon_{i}^{2}} \log \frac{1.25}{\delta} - \sigma_{m}^{2}$, and $U_{i} = |h_{i}|^{2} \beta_{i} P_{i}$. 
As seen in Fig. \ref{solution_optimization}, we first rank the left-over powers from the users after aligning the gradients, i.e., $\{\lambda_{k}\}_{k=1}^{K} $ in an ascending order. We then allocate the powers $Z_{k}$ such that a subset of users $S$ satisfies $\sum_{k=1}^{S} Z_{k} \geq \psi, S \leq K$, to satisfy privacy constraints. This completes the proof of Theorem \ref{optimized_convergence}.
\end{proof}

\begin{figure}[t]
    \centering
    \includegraphics[width=0.55\linewidth]{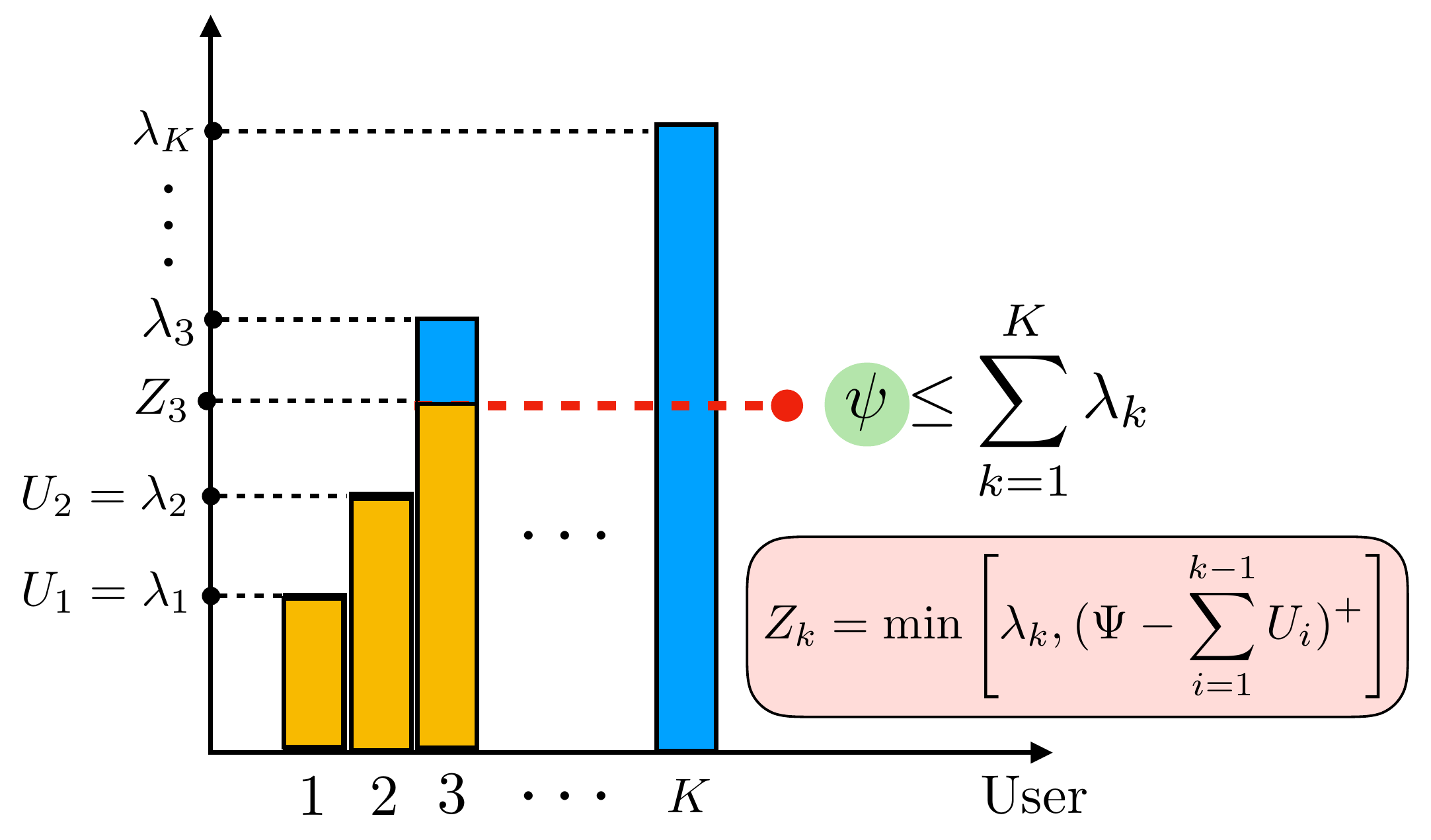}
\vspace{-5pt}
    \caption{An example for the iterative solution: $Z_{1} + Z_{2} + Z_{3} \geq \Psi$, $Z_{k} = 0, k = 4, \cdots, K$. }
    \label{solution_optimization}
    \vspace{-10pt}
\end{figure}

\vspace{-2pt}
\section{Simulation Results} \label{experiments}
\vspace{-1pt}

In this Section, we provide some simulation results to assess the performance of private wireless FL model. We consider a linear regression task on a synthetic dataset. The regularized loss function at the $k$th user is given as: 
\begin{align}
    f_{k}(\mathbf{w}) & = \frac{1}{|\mathcal{D}_k|} \sum_{i=1}^{|\mathcal{D}_k|}(\mathbf{w}^{T} \mathbf{u}_{i}^{(k)} - v_{i}^{(k)})^{2} + \frac{\lambda}{2} \| \mathbf{w} \|^{2}.
\end{align}
Our synthetic dataset consists of 3000  i.i.d. samples drawn from $\mathcal{N}(0, \mathbf{I}_{d+1})$, where  $\mathbf{u}_{i}^{(k)} \in \mathds{R}^{d}$, $v_{i}^{(k)} \in \mathds{R}$ and $d=30$.
We assume that each user has  $|\mathcal{D}_k|=20$ data points.  For the GD algorithm, the regularization parameter $\lambda$ is $10^{-3}$ and $T = 1000$ training iterations. The channel coefficients are drawn from $\mathcal{CN}(0, 1)$, and the channel noise variance is set to $\sigma_{m}^{2} =1$. Also, we assume that each user requires the same privacy level $(\epsilon, \delta) = (1.2, 10^{-4})$-LDP.

 In Fig. \ref{fig:privacy_comparison}(a), we show the impact of the number of users on the training loss for $P_k = 30$ dBm for all $k$.  As we increase the number of users, the training loss decays faster with $T$. In Fig. \ref{fig:privacy_comparison}(b), we compare with the private orthogonal scheme for $KT_{2} = T_{1} = T $ iterations and  $P_k = 30$ dBm for all $k$. Interestingly, the non-orthogonal scheme is more efficient in terms of the bandwidth and accuracy.  In Fig. \ref{fig:privacy_comparison}(c), we show the impact of the transmit power on the training loss where the error decays faster with $T$ as we increase the transmit power.

\begin{figure*}[t]
    \centering
    \includegraphics[width=6in, height=2in]{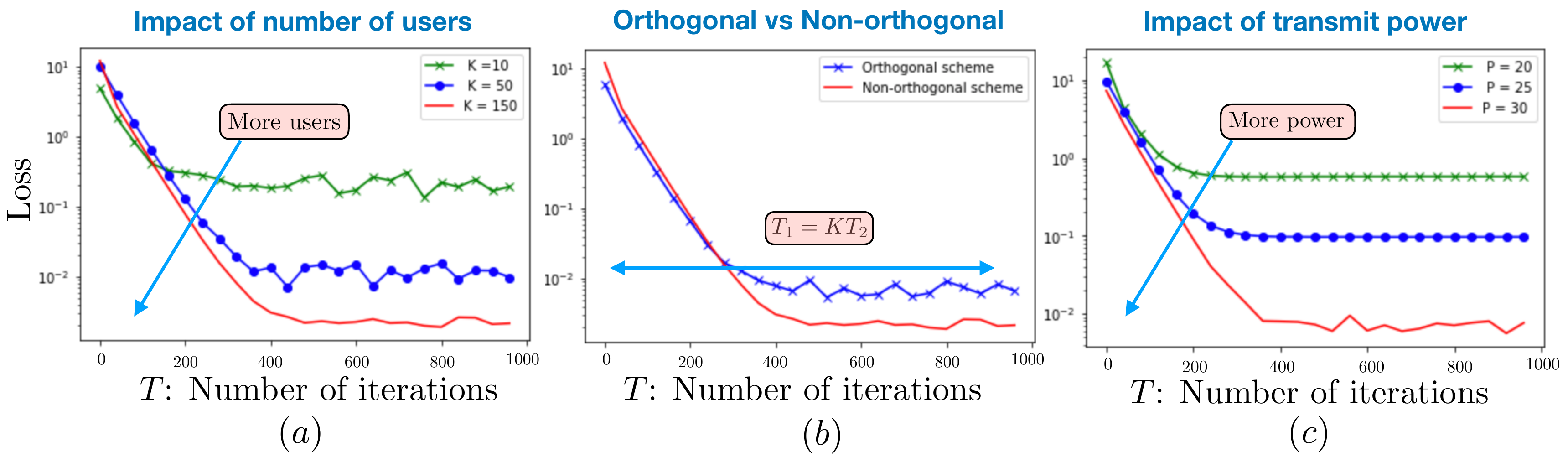}
    \vspace{-5pt}
    \caption{Impact of a) number of users, b) orthogonal vs non-orthogonal transmission, and c) transmit power, on the training loss as a function of  iterations. As we see from the figures, as $T$ increases, the variance term due to the local privacy perturbation and the noisy channel becomes dominant.}
    \label{fig:privacy_comparison}
    \vspace{-5pt}
\end{figure*}

\vspace{-2pt}
\section{Conclusion \& Future Directions}
\label{conclusion}
\vspace{-2pt}
We studied the problem of wireless federated learning   subject to local differential privacy (LDP) constraints. We showed that the wireless channel provides a dual benefit of bandwidth efficiency together with strong LDP guarantees. Using the proposed wireless aggregation scheme, privacy leakage was shown to scale as $\mathcal{O}\big(\frac{1}{\sqrt{K}} \big)$ compared to orthogonal transmission in which the privacy leakage scales as a constant. We also analyzed and optimized the convergence rate of the proposed private FL training algorithm and studied the tradeoffs between wireless resources, convergence, and privacy.

There are several interesting directions for future work, such as generalization to multiple-antennas at the users and the PS. In the proposed scheme, all users align their gradients, which limits the effective SNR by a user with the worst channel conditions. A possible direction would be to explore generalizations of this scheme, by  selecting and aligning gradients from a smaller subsets of users.

\vspace{-2pt}
\section*{Appendix I: Proof of Theorem  \ref{theorem_1_convergence}}\label{proof_theorem_1}
\vspace{-2pt}
To prove  the convergence rate of the proposed algorithm, we recall that the gradient estimate at the PS in (\ref{eq:postpreprocessing}) satisfies: (a) Unbiasedness, i.e., $\mathds{E} \left[\hat{\mathbf{g}}_{t} \right] = \mathds{E} \left[ \nabla F(\mathbf{w}_{t}) \right]$, since the total additive noise is zero mean; and (b) Bounded second moment,  $\mathds{E}\left[\|\hat{\mathbf{g}}_{t}\|_{2}^{2}\right] \leq G^{2}$, which we prove as follows:
\begin{align}
     \mathds{E}\left[\|\hat{\mathbf{g}}_{t}\|_{2}^{2}\right]  & = \mathds{E}\left[\|\nabla F(\mathbf{w}_{t}) + \mathbf{z}_{t}\|_{2}^{2}\right] \nonumber \\ 
    &  =  \mathds{E}\left[\|\nabla F(\mathbf{w}_{t})\|_{2}^{2}\right] + 2 { \mathds{E}\left[\nabla F(\mathbf{w}_{t})^{T} \mathbf{z}_{t}\right] } +  \mathds{E}\left[\|\mathbf{z}_{t}\|_{2}^{2}\right] \nonumber \\ 
    & \overset{(a)} = \|\nabla F(\mathbf{w}_{t})\|_{2}^{2} +  \mathds{E}\left[\|\mathbf{z}_{t}\|_{2}^{2}\right]   \nonumber \\
     & \overset{(b)} \leq   \frac{1}{K^{2}} \times   \bigg( \sum_{k=1}^{K} \| \mathbf{g}_{k}(\mathbf{w}_{t})\|_{2}\bigg)^{2} +  \mathds{E}\left[\|\mathbf{z}_{t}\|_{2}^{2}\right]    \nonumber \\
      & \overset{(c)}  \leq   \frac{1}{K^{2}} \times  ( K L)^{2} +  \frac{d}{K^{2} c^{2}} \ \left[ \sum_{k=1}^{K} |h_{k}|^{2} \beta_{k} P_{k} +1 \right]   \nonumber \\
    & \leq {L^{2}}   + \frac{d}{K^{2} c^{2}} \left[ \sum_{k=1}^{K} |h_{k}|^{2} \beta_{k} P_{k} +1\right] \triangleq G^{2} \label{eq:finalconv}
\end{align}
where (a) follows from the fact that ${ \mathds{E}\left[\nabla F(\mathbf{w}_{t})^{T} \mathbf{z}_{t}\right] = 0 } $,  (b) follows from Cauchy-Schwarz inequality, and (c) from the assumption that $ \| \mathbf{g}_{k}(\mathbf{w}_{t})\|_{2} \leq L$, i.e., the Lipschitz constant $\forall k$. We next invoke standard results \cite{rakhlin2012making} on convergence of SGD for $\mu$-smooth and $\lambda$-strongly convex loss, which states 
\begin{align}
     \mathds{E} \left[ F(\mathbf{w}_{T}) \right] - F(\mathbf{w}^{*}) \leq \frac{2 \mu G^{2}}{\lambda^{2} T}.\label{eq:finalconv2}
\end{align}
Plugging $G^{2}$ from \eqref{eq:finalconv} in \eqref{eq:finalconv2}, we arrive at Theorem \ref{theorem_1_convergence}.

 \bibliographystyle{IEEEtran}
\bibliography{myreferences}

\end{document}